\documentclass[12pt]{article}

\usepackage{amssymb}%
\usepackage{amsmath}%
\usepackage{amsthm}%
\usepackage{url}
\usepackage{tikz}
\usetikzlibrary{arrows.meta,automata,positioning,quotes}

\allowdisplaybreaks%

{\theoremstyle{plain}%
 \newtheorem{theorem}{Theorem}

}
{\theoremstyle{remark}
\newtheorem{remark}{Remark}
}
{\theoremstyle{definition}
\newtheorem{definition}{Definition}
\newtheorem{example}{Example}
}

\begin{document}

\begin{center}
{\Large A generalization of Deterministic Finite Automata related to discharging}

 \ 

{\textsc{John M. Campbell}} 

 \ 

\end{center}

\begin{abstract}
 Deterministic Finite Automata (DFAs) are of central importance in automata theory. In view of how state diagrams for DFAs are 
 defined using directed graphs, this leads us to introduce a generalization of DFAs related to a method widely used in graph theory 
 referred to as the \emph{discharging method}. Given a DFA $(Q, \Sigma, \delta, q_{0}, F)$, the transition function $\delta\colon Q \times 
 \Sigma \to Q$ determines a directed path in the corresponding state diagram based on an input string $a_{1} a_{2} \cdots a_{n}$ 
 consisting of characters in $\Sigma$, and our generalization can be thought of as being based on how each vertex in $D$ ``discharges'' 
 rational values to adjacent vertices (by analogy with the discharging method) depending on the string $a_{1} a_{2} \cdots a_{n}$ and 
 according to a fixed set of rules. We formalize this notion and pursue an exploration of the notion of a \emph{Discharging 
 Deterministic Finite Automaton} (DDFA) introduced in this paper. Our DDFA construction gives rise to a ring structure consisting of 
 sequences that we refer to as being \emph{quasi-$k$-regular}, and this ring generalizes the ring of $k$-regular sequences introduced 
 by Allouche and Shallit. 
\end{abstract}


\vspace{0.1in}


\section{Introduction}\label{sectionIntro}
 Automata theory is of core importance within theoretical computer science. Among the core objects of study within automata theory are 
 finite-state machines, and Deterministic Finite Automata (DFAs) provide an important class of finite state machines. The roles played 
 by DFAs in both the number-theoretic study of automatic sequences and in combinatorics on words lead us to consider generalizing 
 DFAs in ways that are motivated by number-theoretic and combinatorial applications. This has inspired us to introduce a 
 generalization of DFAs that is related to the famous \emph{discharging method} used in graph theory, with reference to Cranston and 
 West's expository work on this method \cite{CranstonWest2017}. 

 There have been many generalizations or variants of or related to DFAs in the context of many different applications. 
 \emph{Nondeterministic Finite Automata}, which have state diagrams that allow multiple directed edges with the same label starting 
 with the same vertex, provide a notable generalization of DFAs, as is the case with \emph{Generalized Automata}, for which each 
 transition is defined on words (or strings or blocks), as opposed to individual letters (or the null string), with Generalized Automata 
 having been introduced within Eilenberg's monographs on automata, languages, and machines \cite{Eilenberg1974}. Brzozowski and 
 McCluskey, Jr., in 1963 \cite{BrzozowskiMcCluskey1963}, introduced finite state automata allowing a regular expression for each 
 transition, thus generalizing Generalized Automata. We also invite the interested reader to review further research contributions on 
 or related to generalizations or variants of DFAs that have inspired our work 
 \cite{GiammarresiMontalbano1999,HanWood2005,Hashigushi1991,Nagy2023}. 

 Since our work is based on generalizing the concept of a DFA, we proceed with the following standard definition, referring the 
 interested reader to background material in the standard monograph on automatic sequences \cite{AlloucheShallit2003}. For further 
 background material concerning DFAs, we refer to the appropriate sections of texts on the theory of computation 
 \cite[\S2.2]{Wood1987} and on formal languages \cite[\S2.21]{RozenbergSalomaa1997}. 

\begin{definition}\label{definitionDFA}
 A \emph{Deterministic Finite Automaton} is a $5$-tuple $(Q, \Sigma, \delta, q_0, F)$ consisting of a finite nonempty set $Q$ of 
 internal \emph{states}, a finite nonempty set $\Sigma$ (an \emph{input alphabet}), a \emph{transition function} $\delta\colon Q 
 \times \Sigma \to Q$, an \emph{initial state} $q_{0} \in Q$, and a set $F \subseteq Q$ of \emph{accepting 
 states} \cite[\S4.1]{AlloucheShallit2003}. 
\end{definition}

 The graph-theoretic intuition behind our generalization of Definition \ref{definitionDFA} has to do with the use of transition diagrams to 
 illustrate DFAs. For a DFA $(Q, \Sigma, \delta, q_0, F)$, and for a nonempty string $w$ consisting of characters in $\Sigma$, we adopt 
 the convention whereby the given DFA takes $w$ as input according to the sequence of characters in $w$ from left to right, again 
 referring to the standard text on automatic sequences for details \cite[\S4.1]{AlloucheShallit2003}. A \emph{transition diagram} or 
 \emph{state diagram} for the DFA $(Q, \Sigma, \delta, q_0, F)$ is a directed graph with $Q$ as its vertex set and with directed edges 
 determined in the following manner. For states or vertices $q$ and $q'$ in $Q$, and for a character $c \in \Sigma$, we have 
 that $\delta(q, c) = q'$ if and only if there is a directed edge from $q$ to $q'$ labeled with $c$ in the state diagram for the 
 given DFA (with vertices in $F$ being denoted with an extra circle). 

\begin{example}\label{exTM}
 In Section \ref{sectiongeneralizeDFA} below, we consider a variant of the state diagram 
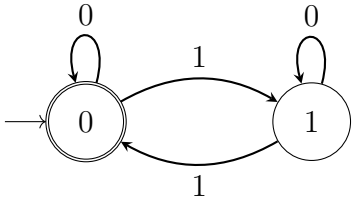
\begin{center}
\begin{tikzpicture} [draw=black!80!black,
 node distance = 3cm, 
 on grid, 
 auto,
 every loop/.style={stealth-}]
 
% State q0 
\node (q0) [state, 
 initial, 
 accepting, 
 initial text = {}] {$0$};
 
% State q1 
\node (q1) [state,
 right = of q0] {$1$};
 
% Arrows
\path [-stealth, thick]
 (q0) edge[bend left] node {$1$} (q1)
 (q1) edge[bend left] node {$1$} (q0)
 (q0) edge [loop above] node {$0$}()
 (q1) edge [loop above] node {$0$}();
\end{tikzpicture}. 
\end{center}
 Setting $Q = \{ q_{0}, q_{1} \}$, with $q_{0} = 0$ and $q_{1} = 1$ and setting $\Sigma = \{ 0, 1 \}$, we define a transition function 
 $\delta\colon Q \times \Sigma \to Q$ so that $\delta(q_0, 0) = 0$, $\delta(q_0, 1) = 1$, $\delta(q_1, 0) = 1$, and $\delta(q_1, 
 1) = 0$. Setting $F = \{ q_{0} \}$, the DFA $(Q, \Sigma, \delta, q_0, F)$ is illustrated with the transition diagram shown above. 
 Inputting the binary expansion for a nonnegative integer $n$ into the specified DFA, the label of the final state associated with this 
 input is the number $\text{{\bf t}}(n)$ of $1$-digits, modulo 2, in the base-$2$ expansion of $n$, with the integer sequence 
\begin{equation}\label{TMsequence}
 (\text{{\bf t}}(n) : n \in \mathbb{N}_{0}) = (0, 1, 1, 0, 1, 0, 0, 1, 1, 0, 0, 1, 0, 1, 1, 0, 1, 0, 0, 1, 0, 1, \ldots) 
\end{equation}
 being known as the \emph{Thue--Morse sequence}. 
\end{example}

 The importance of the computational model given by DFAs in conjunction with the famous ubiquity of the integer sequence in 
 \eqref{TMsequence} \cite{AlloucheShallit1999} lead us to consider generalizations of Example \ref{exTM}. Informally, if we imagine 
 inputting a ``unit charge'' into the state diagram in Example \ref{exTM}, then our generalization can be thought of as being based on 
 how a vertex would discharge according to a fixed set of rules, by analogy with the discharging method described below. 

 The discharging method may be seen as one of the most famous methods in graph theory, in view of how it is involved in the proof of 
 the Four Color Theorem. In its most basic form, given a graph $G$, each vertex $v$ is assigned a ``charge'' equal to the degree of $v$, 
 and portions of such charges are repositioned (typically to adjacent vertices) according to discharging rules to demonstrate the 
 existence of a local structure of a graph \cite{CranstonWest2017}. More specifically, in the application of this simplest form of 
 discharging, under the assumption that $G$ does not involve a structure such as a subgraph or an edge/vertex satisfying certain 
 properties, then, ideally, discharging rules applied to reallocate portions of the initially assigned charge could be used to demonstrate 
 that the final charge for every vertex is at least some specified value in such a way so that this would result in a contradiction (given a 
 known value for the average degree, and since the total charge is constant). 

\section{A generalization of DFAs related to discharging}\label{sectiongeneralizeDFA}
 Ideally, by extending, as below, the definition of a DFA in a way that is motivated by automatic sequences and graph-theoretic 
 properties associated with state diagrams, this could help to give light to related problems in the field of combinatorics on words, or 
 active areas of research related to DFAs, such as research areas related to the conversion of DFAs into regular expressions.

\begin{definition}\label{definitionDDFA}
 A \emph{Discharging Deterministic Finite Automaton} (DDFA) is a $6$-tuple 
 $(Q, \Sigma, \delta, q_0, F, P)$ 
 such that $(Q, \Sigma, \delta, q_0, F)$ is a DFA, and where $P$ is defined as follows. 
 For $q \in Q$ and for $s \in \Sigma$, 
 we let expressions of the forms 
 $n^{\text{current}}_{q, s}$ and 
 $n_{q, s}^{\text{not current}}$ denote fixed nonnegative rational numbers satisfying 
\begin{equation}\label{totalaxiom}
 n^{\text{current}}_{q, s} + 
 \sum_{t \neq s} n_{q, t}^{\text{not current}} = 1, 
\end{equation}
 letting it be understood that the sum over $t \neq s$ in \eqref{totalaxiom}
 is over all characters $t \in \Sigma \setminus \{ s \}$. 
 We then let 
 $$ P = \{ \{ (q, s, t, 
 n^{\text{current}}_{q, s}, n^{\text{not current}}_{q, t} 
 ) : t \in \Sigma \setminus \{ s \} \} : q \in Q, s \in \Sigma \}. $$ 
\end{definition}

\begin{example}\label{exTMlike}
 Let the entries of the tuple 
 $(Q, \Sigma, \delta, q_0, F)$ 
 be as in Example \ref{exTM}, and we set 
\begin{align*}
 n_{q_{0}, 0}^{\text{current}} & = n_{q_{0}, 1}^{\text{not current}} = \frac{1}{2}, \\
 n_{q_{0}, 1}^{\text{current}} & = n_{q_{0}, 0}^{\text{not current}} = \frac{1}{2}, \\ 
 n_{q_{1}, 0}^{\text{current}} & = n_{q_{1}, 1}^{\text{not current}} = \frac{1}{2}, \\ 
 n_{q_{1}, 1}^{\text{current}} & = n_{q_{1}, 0}^{\text{not current}} = \frac{1}{2}. 
\end{align*}
 We let this be illustrated as below, where the two ``$\frac{1}{2}$'' symbols in each vertex 
 are meant to illustrate that, informally, each vertex can be thought of as discharging 
 equally based on the two directed edges starting at this vertex. 
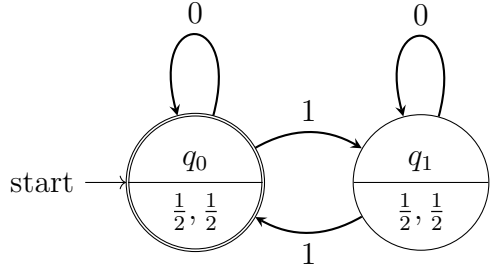
\begin{center} 
\begin{tikzpicture} [draw=black!80!black,
 node distance = 3cm, 
 on grid, 
 auto,
 every loop/.style={stealth-}]
 
% State q0 
\node (q0) [initial,state with output,accepting] {
 $q_0$ \nodepart{lower} $\frac{1}{2}, \frac{1}{2}$
};
 
% State q1 

\node (q1) [state with output, right = of q0] {
 $q_1$ \nodepart{lower} $\frac{1}{2}, \frac{1}{2}$
};

%% Arrows
\path [-stealth, thick]
 (q0) edge[bend left] node {$1$} (q1)
 (q1) edge[bend left] node {$1$} (q0)
 (q0) edge [loop above] node {$0$}()
 (q1) edge [loop above] node {$0$}();
\end{tikzpicture}
\end{center}
\noindent To illustrate how this relates to the discharging method, we begin by assigning a unit charge to the initial state, 
 and we consider 
 inputting the string $1010$ into the above 
 DDFA and keeping track of how the vertices discharge according to the given rules. 
 Inputting the character $1 \in \Sigma$, we find that $\delta(q_0, 1) = q_1$, and we recall that $n_{q_{0}, 1}^{\text{current}} = 
 \frac{1}{2}$. Informally, this can be thought of as signifying that half of the initial unit charge is to be reallocated to the 
 vertex $q_{1}$. Since $n_{q_{0}, 0}^{\text{not current}} = \frac{1}{2}$, this can be thought of as signifying that the remaining half of 
 the initial unit charge (apart from the half-charge sent from $q_0$ to $q_1$) is sent back to $q_{0}$, and this is illustrated below, 
 where the colored vertex indicates the current state according to the relation such that $\delta(q_0, 1) = q_1$ and that the 
 character $1 \in \Sigma$ has been inputted. 
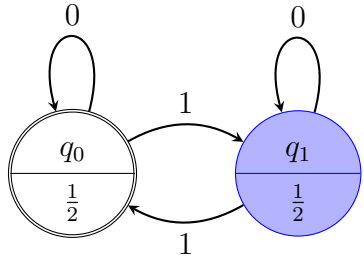
\begin{center}
\begin{tikzpicture} [draw=black!80!black,
 node distance = 3cm, 
 on grid, 
 auto,
 every loop/.style={stealth-}]

\node (q0) [state with output,accepting] {
 $q_0$ \nodepart{lower} $\frac{1}{2}$
};

\node (q1) [state with output, right = of q0,draw = blue, fill = blue!30] {
 $q_1$ \nodepart{lower} $\frac{1}{2}$
};

\path [-stealth, thick]
 (q0) edge[bend left] node {$1$} (q1)
 (q1) edge[bend left] node {$1$} (q0)
 (q0) edge [loop above] node {$0$}()
 (q1) edge [loop above] node {$0$}();
\end{tikzpicture}
\end{center}
\noindent After inputting $1 \in \Sigma$, we then apply the transition rule such that $\delta(q_{1}, 0) = q_{1}$. Informally, given that 
 the charge associated with $q_{1}$, after having inputted the symbol $1$, is $\frac{1}{2}$, since $n_{q_{1}, 0}^{\text{current}} = 
 \frac{1}{2}$, we can think of half of the half-charge associated with $q_{1}$ as being sent from $q_{1}$ and back to itself, with the 
 remaining charge being sent to $q_0$, yielding the charge distribution indicated below, where the coloured vertex indicates the 
 current state given by the relation $\delta(q_{0}, 0) = q_{0}$. 
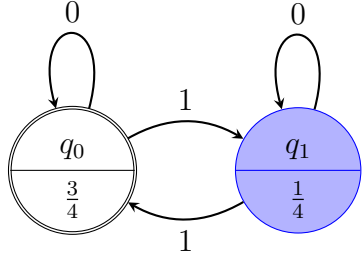
\begin{center}
\begin{tikzpicture} [draw=black!80!black,
 node distance = 3cm, 
 on grid, 
 auto,
 every loop/.style={stealth-}]

\node (q0) [state with output,accepting] {
 $q_0$ \nodepart{lower} $\frac{3}{4}$
};

\node (q1) [state with output, right = of q0,draw = blue, fill = blue!30] {
 $q_1$ \nodepart{lower} $\frac{1}{4}$
};

\path [-stealth, thick]
 (q0) edge[bend left] node {$1$} (q1)
 (q1) edge[bend left] node {$1$} (q0)
 (q0) edge [loop above] node {$0$}()
 (q1) edge [loop above] node {$0$}();
\end{tikzpicture}
\end{center}
\noindent Continuing in the manner suggested above, according to the input string
 $1010$, the next step in the input process may be illustrated with 
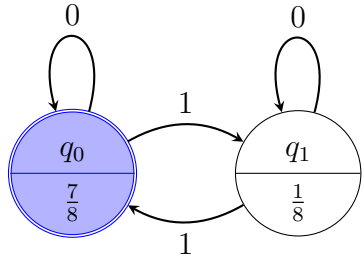
\begin{center}
\begin{tikzpicture} [draw=black!80!black,
 node distance = 3cm, 
 on grid, 
 auto,
 every loop/.style={stealth-}]
 
% State q0 
\node (q0) [state with output,draw = blue, fill = blue!30,accepting] {
 $q_0$ \nodepart{lower} $\frac{7}{8}$
};
 
% State q1 

\node (q1) [state with output, right = of q0] {
 $q_1$ \nodepart{lower} $\frac{1}{8}$
};

%% Arrows
\path [-stealth, thick]
 (q0) edge[bend left] node {$1$} (q1)
 (q1) edge[bend left] node {$1$} (q0)
 (q0) edge [loop above] node {$0$}()
 (q1) edge [loop above] node {$0$}();
\end{tikzpicture}
\end{center}

\noindent and the remaining step is illustrated below. 

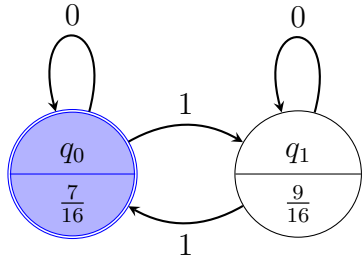
\begin{center}
\begin{tikzpicture} [draw=black!80!black,
 node distance = 3cm, 
 on grid, 
 auto,
 every loop/.style={stealth-}]
 
% State q0 
\node (q0) [state with output,draw = blue, fill = blue!30,accepting] {
 $q_0$ \nodepart{lower} $\frac{7}{16}$
};
 
% State q1 

\node (q1) [state with output, right = of q0] {
 $q_1$ \nodepart{lower} $\frac{9}{16}$
};

%% Arrows
\path [-stealth, thick]
 (q0) edge[bend left] node {$1$} (q1)
 (q1) edge[bend left] node {$1$} (q0)
 (q0) edge [loop above] node {$0$}()
 (q1) edge [loop above] node {$0$}();
\end{tikzpicture}
\end{center}
\end{example}

 By taking the same DDFA as in Example \ref{exTMlike}, and by mimicking the steps involved in Example \ref{exTMlike} based on input 
 strings given by the base-2 expansions for consecutive nonnegative integers, and by then recording the final charge $a_{n}$ attained by 
 the final state (which may be thought of as the last colored vertex by analogy with the illustrations in Example \ref{exTMlike}), this 
 gives us the sequence 
\begin{equation}\label{eqTMlike}
 (a_{n} : n \in \mathbb{N}_{0}) = \left( \frac{1}{2},\frac{1}{2},\frac{1}{4},\frac{3}{4}, \frac{1}{8}, \frac{7}{8}, 
 \frac{3}{8},\frac{5}{8},\frac{1}{16},\frac{15}{16},\frac{7}{16}, 
 \frac{9}{16},\frac{3}{16},\frac{13}{16},\frac{5}{16}, \ldots \right). 
\end{equation}
 For example, the valuation $a_{10} = \frac{7}{16}$ shown in \eqref{eqTMlike} 
 agrees with the final charge computed in Example \ref{exTMlike}. 
 By taking the integer sequence 
\begin{equation}\label{seqnotinOEIS}
 (1, 1, 1, 3, 1, 7, 3, 5, 1, 15, 7, 9, 3, 13, 5, 11, 1, 31, 15, 17, 7, 25, 9, 23, 3, \ldots)
\end{equation}
 given by the sequence of consecutive numerators for the reduced fractions for the entries in \eqref{eqTMlike}, we find that the 
 integer sequence in \eqref{seqnotinOEIS} is not currently included in the {O}n-{L}ine {E}ncyclopedia of {I}nteger {S}equences 
 \cite{Sloane2025}, suggesting that our DDFA construction is new. We may equivalently define the rational sequence in 
 \eqref{eqTMlike} according to the recursion such that 
\begin{equation}\label{ancases}
 a_{n} = \begin{cases} 
 \frac{a\left( \frac{n}{2} \right)}{2} & \text{if $n \equiv 0 \, \text{mod} \, 2$ and $n > 1$,} \\ 
 1 - \frac{a\left( \frac{n-1}{2} \right)}{2} & \text{if $n \equiv 1 \, \text{mod} \, 2$ and $n > 1$,} 
 \end{cases} 
\end{equation}
 with the initial values $a_{0} = a_{1} = \frac{1}{2}$. If we let $b_{n}$, for a nonnegative integer $n$, 
 denote the numerator of the reduced fraction associated with $a_{n}$, then 
 we find that the integer sequence 
\begin{equation}\label{modifybn}
 \left( \frac{b(n) + 1}{2} : n \in \mathbb{N} \right) 
 = (1, 1, 2, 1, 4, 2, 3, 1, 8, 4, 5, 2, 7, 3, 6, 1, 16, 8, 9, \ldots) 
\end{equation}
 agrees (up to an index shift) with the sequence that is indexed as {\tt A131271} 
 in the OEIS and that is defined below. 

 We construct a triangular array with entries $T(n, k)$ for $n \geq 0$ and for $k \in \{ 1, 2, \ldots, 2^{n} \}$, we set $T(0, 1) = 1$ 
 and $T(n, 2k-1) = T(n-1, k)$ and $T(n, 2k) = 2^{n} + 1 - T(n - 1, k)$, and we then form the integer sequence {\tt A131271} by taking 
 consecutive entries among the consecutive rows of the triangular array we have defined. It is immediate from the recursion in 
 \eqref{ancases} that the integer sequence in \eqref{modifybn} agrees with {\tt A131271} (up to an offset). The {\tt A131271} 
 sequence has direct applications in respiratory physiology, with this OEIS sequence being referenced explicitly in the work of 
 Demongeot and Waku on the use of interval iterations in the context of an entrainment problem \cite{DemongeotWaku2009}. 
 It seems that the sequence {\tt A131271} has not previously been considered in relation 
 to DFAs or variants or generalizations of DFAs, which motivates the application of our construction 
 in relation to the work of Demongeot and Waku \cite{DemongeotWaku2009}. 
 As noted in the OEIS entry indexed as {\tt A131271}, 
 the integer sequence described therein agrees with a construction due to 
 Demongeot and Waku \cite{DemongeotWaku2009} involving limit cycles of interval iterations, 
 with applications in the fields of 
 mathematical demography and mathematical population studies \cite{DemongeotWaku2005}. 
 This further motivates the notion of a DDFA we introduce in this paper, 
 in view of the connection between our Thue--Morse-like sequence in 
 \eqref{eqTMlike} and the research due to 
 Demongeot and Waku \cite{DemongeotWaku2009,DemongeotWaku2005}. 

\subsection{Discharging Deterministic Finite Automata with Output}
 Our discharging-based generalization of DFAs 
 is inspired by past graph-theoretic approaches toward or related to the study of automata, 
 as in the work of Restivo and Vaglica \cite{RestivoVaglica2012} and 
 further references 
 \cite{BrandenburgSkodinis2005,FraigniaudIlcinkasPeerPelcPeleg2005,Gruber2012,Kelarev2004,KonitzerSimon2017}. 
 To begin with, we require Definition \ref{definitionDFAO} below for the purposes of our construction. 

\begin{definition}\label{definitionDFAO}
 A \emph{Deterministic Finite Automaton with Output} (DFAO) 
 is a $6$-tuple $(Q, \Sigma, \delta, q_0, \Delta, \tau)$, 
 for $Q$ and $\Sigma$ and $\delta$ and $q_0$ as in Definition \ref{definitionDFA}, 
 and where $\Delta$ is an \emph{output alphabet}, with $\tau\colon Q \to \Delta$ as the 
 \emph{output function} \cite[p.\ 138]{AlloucheShallit2003}. 
\end{definition}

 Let $\Sigma^{\ast}$ denote the free monoid on $\Sigma$. For a DFA $(Q, \Sigma, \delta, q_0, F)$, the domain of the function 
 $\delta$ is extended from $Q \times \Sigma$ to $Q \times \Sigma^{\ast}$ as follows \cite[p.\ 129]{AlloucheShallit2003}. For the 
 empty string $\epsilon \in \Sigma^{\ast}$, and for arbitrary $q \in Q$, we set $\delta(q, \epsilon) = q$. We then define $\delta$ 
 recursively so that $\delta(q, xa) = \delta(\delta(q, x), a)$ for $x \in \Sigma^{\ast}$ and $a \in \Sigma$. Informally, if we consider 
 inputting the consecutive characters of a word $w$ in $\Sigma^{\ast}$ into a state diagram associated with a DFA, then, by moving 
 from state to state (according to the transition function for the given DFA and according to the consecutive characters in $w$), we 
 would take the label of the final vertex and then apply the mapping $\tau$ to this label, thereby providing a generalization of DFAs in 
 the sense that a DFA may be thought of as computing a function from $\Sigma^{\ast}$ to $\{ 0, 1 \}$, i.e., depending on whether an 
 accepting state is reached \cite[p.\ 139]{AlloucheShallit2003}. The many applications of DFAs, including in areas of protocol analysis 
 and natural language processing, motivate our further generalizing DFAs. 

\begin{definition}
 A \emph{Discharging Deterministic Finite Automaton with Output} (DDFAO) is a $7$-tuple 
 $(Q, \Sigma, \delta, q_0, P, \Delta, \tau)$ 
 for $Q$ and $\Sigma$ and $\delta$ and $q_0$ and $P$ as in Definition \ref{definitionDDFA}, 
 again letting $\Delta$ denote an output alphabet, 
 and again letting $\tau\colon Q \to \Delta$ 
 be an output function. 
\end{definition}

 In view of how the transition function for a DFA 
 $(Q, \Sigma, \delta, q_0, F)$ is extended from $Q \times \Sigma$ to $Q \times \Sigma^{\ast}$ 
 (and this resultant mapping is often referred to as an \emph{extended transition function} and is often denoted with $\delta^{\ast}$), 
 we introduce a further extension, as below, that, informally, takes into account the notion of ``charge'' 
 implicit in our construction of a DDFAO. 

\begin{definition}\label{definitionchargeext}
 Letting $Q$ and $\Sigma$ and $\delta$ and $q_0$ and $P$
 be as in the definition for a DDFA, 
 we define the 
 \emph{charge-extended transition function}
 $$ \delta^{c}\colon Q \times \Sigma^{\ast} \to Q \times [0, 1] $$
 as follows. For distinguished $q \in Q$, we set $\delta^{c}(q, \epsilon) = (q, 1)$. 
 For a nonempty word $w \in \Sigma^{\ast}$ we write $w = w_{1} w_{2} \cdots w_{\ell(w)}$. 
 We also let $Q = \{ q^{(1)}$, $ q^{(2)}$, $ \ldots$, $ q^{(|Q|)} \}$, 
 and we define the \emph{contemporary charge} $c^{q, i}$ of an element 
 in $Q$ and parameterized by $q$ 
 and $i \in \{ 0, 1, 2, \ldots, \ell(w) \}$ as follows. 
 To begin with, we set $c^{q, 0}(q) = 1$ 
 and $c^{q, 0}(q^{(j)}) = 0$ for $q^{(j)} \neq q$. 
 For $i \in \{ 1, 2, 3, \ldots, \ell(w) \}$, we then let $k_i$ denote the unique 
 integer in $\{ 1, 2, \ldots, |Q| \}$ such that 
\begin{equation}\label{deltanested}
 q^{(k_i)} = \delta(\ldots \delta(\delta(\delta(q, w_1), w_2), w_3), \ldots, w_{i}), 
\end{equation}
 and we also write $q^{(k_0)} = q$. 
 Adopting the convention whereby a sum over the null set is $0$, we set 
\begin{multline*}
 c^{q, i+1}\big( q^{(k_{i})} \big) = \\
 \begin{cases} 
 c^{q, i}\big( q^{(k_{i})} \big) \left( 1 - n^{\text{current}}_{q^{(k_i)}, w_{i+1}} 
 - \sum_{ \substack{ t \neq w_{i+1} \\ \delta(q^{(k_i)}, t) \neq q^{(k_i)} } } n_{q^{(k_i)}, t}^{\text{not current}} \right) & \null \\ 
 \hspace{2.7in} \text{if $\delta\big( q^{(k_i)}, w_{i+1} \big) \neq q^{(k_i)} $}, & \null \\
 c^{q, i}\big( q^{(k_{i})} \big) \left( 1 
 - \sum_{ \substack{ t \neq w_{i+1} \\ \delta(q^{(k_i)}, t) \neq q^{(k_i)} } } n_{q^{(k_i)}, t}^{\text{not current}} \right) & \null \\ 
 \hspace{2.7in} \text{if $\delta\big( q^{(k_i)}, w_{i+1} \big) = q^{(k_i)} $}, & \null 
 \end{cases}
\end{multline*}
 with 
\begin{multline*}
 c^{q, i+1}\big( q^{(k_{i+1})} \big) = \\
 \begin{cases} 
 c^{q, i}\big( q^{(k_{i+1})} \big) + c^{q, i}\big( q^{(k_{i})} \big) \left( n^{\text{current}}_{q^{(k_i)}, w_{i+1}} 
 + \sum_{ \substack{ t \neq w_{i+1} \\ \delta(q^{(k_i)}, t) = q^{(k_{i+1})} } } n_{q^{(k_i)}, t}^{\text{not current}} \right) & \null \\ 
 \hspace{2.7in} \text{if $\delta\big( q^{(k_i)}, w_{i+1} \big) \neq q^{(k_i)} $}, & \null \\ 
 c^{q, i+1}\big( q^{(k_i)} \big) & \null \\ 
 \hspace{2.7in} \text{if $\delta\big( q^{(k_i)}, w_{i+1} \big) = q^{(k_i)} $}, & \null 
 \end{cases} 
\end{multline*}
 and, for $j \in \{ 1, 2, \ldots, |Q| \}$ not equal to $k_i$ and not equal to $k_{i+1}$, we set 
\begin{equation}\label{202506121177A77M1A}
 c^{q, i+1}\big( q^{(j)} \big) = c^{q, i}\big( q^{(j)} \big) 
 + c^{q, i}\big( q^{(k_{i})} \big) \sum_{\substack{t \\ \delta\left( q^{k_i}, t \right) 
 = q^{(j)} }} n_{q^{(k_i)}, t}^{\text{not current}}. 
\end{equation}
 We then set 
\begin{equation}\label{setdeltacqw}
 \delta^{c}( q, w ) 
 = \left( q^{( k_{\ell(w)} )}, 
 c^{q, \ell(w)}\big( q^{( k_{\ell(w)} )} \big) \right). 
\end{equation}
 We also write $ \delta^{c}( q, w, 1 ) = q^{( k_{\ell(w)} )}$ and 
 $ \delta^{c}( q, w, 2 ) = 
 c^{q, \ell(w)}\big( q^{( k_{\ell(w)} )} \big)$. 
\end{definition}

\begin{example}\label{extuplecharge}
 Let $Q = \{ q_0, q_1, q_2, q_3 \}$, let $\Sigma = \{ 0, 1 \}$, 
 and let the transition function $\delta\colon Q \times \Sigma \to Q$ be as illustrated in the diagram below. 
 As in Example \ref{exTMlike}, 
 the copies of $\frac{1}{2}$ within the vertices of the following diagram are meant to illustrate that, informally, 
 each vertex discharges evenly according to the two directed edges starting from the same vertex (regardless of the ``current''
 or ``not current'' superscripts). 
 Explicitly, we have that 
\begin{align*}
 n_{q_{i}, 0}^{\text{current}} & = n_{q_{i}, 1}^{\text{not current}} = \frac{1}{2} \ \text{and} \\
 n_{q_{i}, 1}^{\text{current}} & = n_{q_{i}, 0}^{\text{not current}} = \frac{1}{2} 
\end{align*} 
 for indices $i \in \{ 0, 1, 2, 3 \}$, 
 and this is illustrated below. 
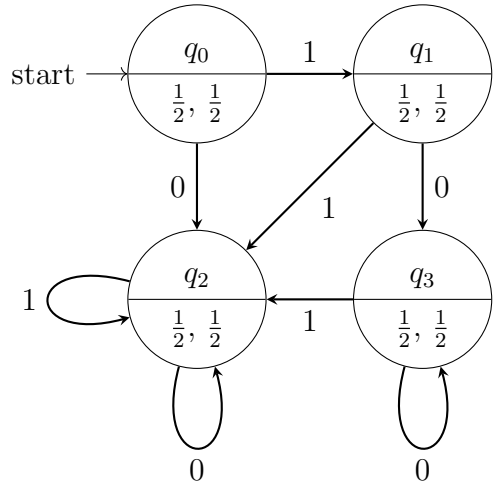
\begin{center}
\begin{tikzpicture} [draw=black!80!black,
 node distance = 3cm, 
 on grid, 
 auto,
 every loop/.style={stealth-}]
 
% State q0 
\node (q0) [initial,state with output] {
 $q_0$ \nodepart{lower} $\frac{1}{2}$, $\frac{1}{2}$
};
 
% State q1 

\node (q1) [state with output, right = of q0] {
 $q_1$ \nodepart{lower} $\frac{1}{2}$, $\frac{1}{2}$
};

\node (q2) [state with output, below = of q0] {
 $q_2$ \nodepart{lower} $\frac{1}{2}$, $\frac{1}{2}$
};

\node (q3) [state with output, right = of q2] {
 $q_3$ \nodepart{lower} $\frac{1}{2}$, $\frac{1}{2}$
};

\path [-stealth, thick]
 (q0) edge node {$1$} (q1)
 (q0) edge[swap] node {$0$} (q2)
 (q1) edge node {$0$} (q3)
 (q1) edge node {$1$} (q2)
 (q3) edge node {$1$} (q2)
 (q2) edge [loop below] node {$0$}()
 (q2) edge [loop left] node {$1$}()
 (q3) edge [loop below] node {$0$}();
\end{tikzpicture}
\end{center}
 This diagram is related to the automatic sequence known as the Fredholm--Rueppel sequence given by the characteristic function for 
 positive powers of $2$, and this is later clarified. 
 Informally, by inputting, for example, the word $w = 1010$ into 
 the DDFA illustrated above, 
 we can mimic the steps illustrated in Example \ref{exTMlike} 
 to compute the final charge $\delta^{c}(q_0, 1010)$
 associated with the given input. 
 From \eqref{deltanested}, we have that 
\begin{equation*}
 q^{(k_1)} = \delta(q_0, w_1) = \delta(q_0, 1) = q_{1}. 
\end{equation*}
 Since $\delta\big( q^{(k_i)}, w_{i + 1} \big) \neq q^{(k_i)}$ 
 for the $i = 0$ case, 
 i.e., since 
 $\delta\big( q_{0}, 1 \big) = q_{1} \neq q_{0}$,
 we find that 
\begin{align*}
 c^{q_0, 1}\big( q_{0} \big) 
 & = c^{q_0, 0}\big( q_{0} \big) \left( 1 - n^{\text{current}}_{q_{0}, 1} 
 - n_{q_{0}, 0}^{\text{not current}} \right) = 0, 
\end{align*}
 according to the case of the recursion in Definition \ref{definitionchargeext}
 for evaluating $c^{q, i+1}\big( q^{(k_i)} \big)$
 if $\delta\big( q^{(k_i)}, w_{i+1} \big) \neq q^{(k_i)}$. 
 Similarly, we find that 
 $$ c^{q_0, 1}\big( q_1 \big) = 
 c^{q_0, 0}\big( q_1 \big) + c^{q_0, 0}\big( q_0 \big) \left( n^{\text{current}}_{q_0, 1} 
 + \sum_{ \substack{ t \neq 1 \\ \delta(q_0, t) = q_1 } } n_{{q_0}, t}^{\text{not current}} \right) = \frac{1}{2}, $$
 according to the case of the recursion in Definition \ref{definitionchargeext}
 for evaluating
 $c^{q, i+1}\big( q^{(k_{i+1})} \big)$
 if $\delta\big( q^{(k_i)}, w_{i+1} \big) \neq q^{(k_i)}$. Simliarly, from \eqref{202506121177A77M1A}, 
 we obtain that 
\begin{equation*}
 c^{q_0, 1}\big( q_2 \big) = c^{q_0, 0}\big( q_2 \big) 
 + c^{q_0, 0}\big( q_0 \big) \sum_{\substack{t \\ \delta\left( q_0, t \right) 
 = q_2 }} n_{q_0, t}^{\text{not current}} = \frac{1}{2}, 
\end{equation*}
 and another application of 
 \eqref{202506121177A77M1A} gives us that
 $c^{q_0, 1}\big( q_{3} \big) = 0$, yielding the charge distribution illustrated below, 
 where the highlighted vertex illustrates the ``current'' vertex after the application of the transition
 function $\delta$ to obtain $\delta(q_0, 1) = q_1$. 
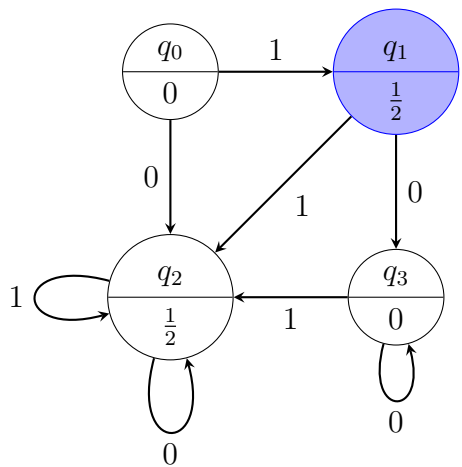
\begin{center}
\begin{tikzpicture} [draw=black!80!black,
 node distance = 3cm, 
 on grid, 
 auto,
 every loop/.style={stealth-}]
 
% State q0 
\node (q0) [state with output] {
 $q_0$ \nodepart{lower} $0$
};
 
% State q1 

\node (q1) [state with output, right = of q0,draw = blue, fill = blue!30] {
 $q_1$ \nodepart{lower} $\frac{1}{2}$
};

\node (q2) [state with output, below = of q0] {
 $q_2$ \nodepart{lower} $\frac{1}{2}$
};

\node (q3) [state with output, right = of q2] {
 $q_3$ \nodepart{lower} $0$
};

\path [-stealth, thick]
 (q0) edge node {$1$} (q1)
 (q0) edge[swap] node {$0$} (q2)
 (q1) edge node {$0$} (q3)
 (q1) edge node {$1$} (q2)
 (q3) edge node {$1$} (q2)
 (q2) edge [loop below] node {$0$}()
 (q2) edge [loop left] node {$1$}()
 (q3) edge [loop below] node {$0$}();
\end{tikzpicture}
\end{center}
A similar approach yields the charge distribution illustrated below, 
 where the highlighted vertex illustrates the ``current'' vertex corresponding to 
 the right-hand side of $\delta(q_1, 0) = q_3$. 
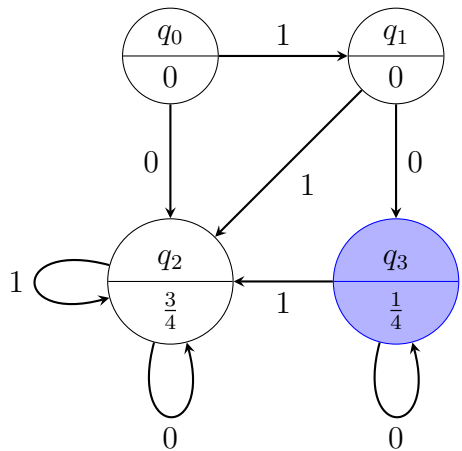
\begin{center}
\begin{tikzpicture} [draw=black!80!black,
 node distance = 3cm, 
 on grid, 
 auto,
 every loop/.style={stealth-}]
 
% State q0 
\node (q0) [state with output] {
 $q_0$ \nodepart{lower} $0$
};
 
% State q1 

\node (q1) [state with output, right = of q0] {
 $q_1$ \nodepart{lower} $0$
};

\node (q2) [state with output, below = of q0] {
 $q_2$ \nodepart{lower} $\frac{3}{4}$
};

\node (q3) [state with output, right = of q2,draw = blue, fill = blue!30] {
 $q_3$ \nodepart{lower} $\frac{1}{4}$
};

\path [-stealth, thick]
 (q0) edge node {$1$} (q1)
 (q0) edge[swap] node {$0$} (q2)
 (q1) edge node {$0$} (q3)
 (q1) edge node {$1$} (q2)
 (q3) edge node {$1$} (q2)
 (q2) edge [loop below] node {$0$}()
 (q2) edge [loop left] node {$1$}()
 (q3) edge [loop below] node {$0$}();
\end{tikzpicture}
\end{center}
\noindent We then obtain the charge distribution illustrated below, 
 which may be thought of as corresponding to the third character $w_3 = 1$ in the input word $w = 1010$. 
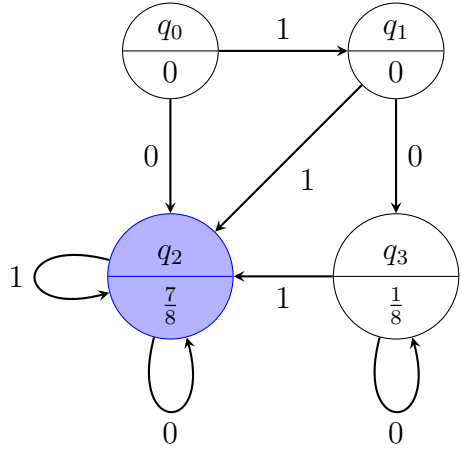
\begin{center}
\begin{tikzpicture} [draw=black!80!black,
 node distance = 3cm, 
 on grid, 
 auto,
 every loop/.style={stealth-}]
 
% State q0 
\node (q0) [state with output] {
 $q_0$ \nodepart{lower} $0$
};
 
% State q1 

\node (q1) [state with output, right = of q0] {
 $q_1$ \nodepart{lower} $0$
};

\node (q2) [state with output, below = of q0,draw = blue, fill = blue!30] {
 $q_2$ \nodepart{lower} $\frac{7}{8}$
};

\node (q3) [state with output, right = of q2] {
 $q_3$ \nodepart{lower} $\frac{1}{8}$
};

\path [-stealth, thick]
 (q0) edge node {$1$} (q1)
 (q0) edge[swap] node {$0$} (q2)
 (q1) edge node {$0$} (q3)
 (q1) edge node {$1$} (q2)
 (q3) edge node {$1$} (q2)
 (q2) edge [loop below] node {$0$}()
 (q2) edge [loop left] node {$1$}()
 (q3) edge [loop below] node {$0$}();
\end{tikzpicture}
\end{center}
 By mimicking previous steps, we then obtain that 
\begin{equation}\label{displaytuplecharge}
 \delta^{c}(q_0, 1010) = \left( q_2, \frac{7}{8} \right). 
\end{equation}
\end{example}

 Since a DFA is the most basic model of a computer \cite[p.\ 4]{Shallit2009}, this motivates, as suggested in Section \ref{sectionIntro}, 
 the application of generalization 
 of DFAs. To formalize how charge-extended transition functions
 can be used to generalize DFAs, we require Definition \ref{reduceddefinition} below. 

\begin{remark}
 For convenience, we adopt the convention whereby any numerical values for elements in $Q$ or $\Sigma$ or $\Delta$ 
 for a DFA/DDFA or DFAO/DDFAO are in $\mathbb{Q}$ (as opposed to, say, $\mathbb{R}$). 
\end{remark}

\begin{definition}\label{reduceddefinition}
 Let $Q$ and $\Sigma$ and $\delta$ and $q_0$ and $P$ be as in the definition of a DDFA, with $\delta^{c}$ as in Definition 
 \ref{definitionchargeext}. We then define the \emph{reduced charge-extended transition function} $\overline{\delta^{c}}$ on $Q 
 \times \Sigma^{\ast}$ so that, for $q \in Q$ and for $w \in \Sigma^{\ast}$, we have that 
 $$\overline{\delta^{c}} = \delta^{c}(q, w, 1) \, \delta^{c}(q, w, 2) $$ 
 if $ \delta^{c}(q, w, 1) \in \mathbb{Q}$, 
 and, otherwise, we write 
 $\overline{\delta^{c}} $ as the formal product 
 of $ \delta^{c}(q, w, 1) $ and $\delta^{c}(q, w, 2)$, denoted as 
 $ \delta^{c}(q, w, 1) \delta^{c}(q, w, 2)$, 
 with the understanding that $q 1 = q$ and $q 0 = 0$ for $q \in Q$. 
 Simlarly, we define the \emph{reduced output function}
 $\overline{\tau}$ so that
 $$ \overline{\tau}\left( \overline{\delta^{c}}\big( q, w \big) \right) 
 = \tau\left( \delta^{c}(q, w, 1) \right) \delta^{c}(q, w, 2), $$
 with the same conventions as before concerning numerical products, and similarly for formal products. 
\end{definition}

\begin{example}
 From the evaluation in \eqref{displaytuplecharge} from Example \ref{extuplecharge}, we obtain the formal product 
\begin{equation*}
 \overline{\delta^{c}}(q_0, 1010) = \frac{7q_2}{8}. 
\end{equation*}
\end{example}

 Our DDFAO construction gives rise to generalizations of automatic and $k$-regular sequences, 
 as we later demonstrate. 
 To begin with, we formalize how 
 our DDFAO (resp.\ DDFA) construction generalizes DFAOs (resp.\ DFAs), as below. 

\begin{theorem}\label{theoremmainformalize}
 Let $(Q, \Sigma, \delta, q_0, P, \Delta, \tau)$ be a DDFAO. For each $q \in Q$ and $s \in \Sigma$, let 
 $ n^{\text{current}}_{q, s} = 1$ and 
 $n_{q, t}^{\text{not current}} = 0$ for $t \in \Sigma \setminus \{ s \}$. Then 
\begin{equation}\label{20250612153PM1A}
 \overline{\delta^{c}}(q_0, w) = \delta(q_0, w) 
\end{equation}
 for an arbitrary word $w \in \Sigma^{\ast}$. 
\end{theorem}

\begin{proof}
 From Definition \ref{definitionchargeext}, we find that 
\begin{equation}\label{20250612202PM1A}
 c^{q_0, i+1}\big( q^{(k_{i})} \big) = 
 \begin{cases} 
 0 & \text{if $\delta\big( q^{(k_i)}, w_{i+1} \big) \neq q^{(k_i)} $}, \\ 
 c^{q_0, i}\big( q^{(k_{i})} \big) 
 & \text{if $\delta\big( q^{(k_i)}, w_{i+1} \big) = q^{(k_i)} $}. 
 \end{cases} 
\end{equation}
 Similarly, Definition \ref{definitionchargeext} gives us that 
\begin{equation}\label{fromdefinitioncases2}
 c^{q_0, i+1}\big( q^{(k_{i+1})} \big) = \\ 
 \begin{cases} 
 c^{q_0, i}\big( q^{(k_{i+1})} \big) + c^{q_0, i}\big( q^{(k_{i})} \big) 
 & \text{if $\delta\big( q^{(k_i)}, w_{i+1} \big) \neq q^{(k_i)} $}, \\ 
 c^{q_0, i+1}\big( q^{(k_i)} \big) & \text{if $\delta\big( q^{(k_i)}, w_{i+1} \big) = q^{(k_i)} $}, 
 \end{cases} 
\end{equation}
 and, letting $j \in \{ 1, 2, \ldots, |Q| \}$ be not equal to $k_i$ and not equal to $k_{i+1}$, we find that 
\begin{equation}\label{2025061241888888888878P8M8A}
 c^{q_0, i+1}\big( q^{(j)} \big) = c^{q_0, i}\big( q^{(j)} \big). 
 \end{equation}
 For the case whereby the empty word $w$ is empty, the desired equality in 
 \eqref{20250612153PM1A} according to the $w = \epsilon$
 cases given in the definitions for an extended transition function and for a (reduced)
 charge-extended transition function. 
 So, suppose that $w$ is nonempty. 
 We claim that 
\begin{equation}\label{claimcq0i}
 c^{q_0, i}\big( q^{(k_i)} \big) = 1
\end{equation}
 and that $c^{q_0, i}\big( q^{(j)} \big) = 0$ for $q^{(j)} \neq q^{(k_i)}$, 
 for indices $i \in \{ 0, 1, \ldots, \ell(w) \}$. 
 From Definition \ref{definitionchargeext}, 
 we find that 
 $c^{q_0, 0}(q_0) = 1$ and $c^{q_0, 0}\big( q^{(j)} \big) = 0$ 
 for $q^{(j)} \neq q_0$. Inductively, we suppose that the given claim holds for for indices $j$ not exceeding $i$. 
 Then \eqref{20250612202PM1A} reduces so that 
\begin{equation}\label{2025111000061020411PM1A}
 c^{q_0, i+1}\big( q^{(k_{i})} \big) = 
 \begin{cases} 
 0 & \text{if $\delta\big( q^{(k_i)}, w_{i+1} \big) \neq q^{(k_i)} $}, \\ 
 1 
 & \text{if $\delta\big( q^{(k_i)}, w_{i+1} \big) = q^{(k_i)} $}, 
 \end{cases} 
\end{equation}
 and \eqref{fromdefinitioncases2} reduces so that 
\begin{equation}\label{2025060122544445444P4M4A}
 c^{q_0, i+1}\big( q^{(k_{i+1})} \big) = \\ 
 \begin{cases} 
 c^{q_0, i}\big( q^{(k_{i+1})} \big) + 1 
 & \text{if $\delta\big( q^{(k_i)}, w_{i+1} \big) \neq q^{(k_i)} $}, \\ 
 c^{q_0, i+1}\big( q^{(k_i)} \big) & \text{if $\delta\big( q^{(k_i)}, w_{i+1} \big) = q^{(k_i)} $}, 
 \end{cases} 
\end{equation}
 and, moreover, since $\delta\big( q^{(k_i)}, w_{i+1} \big) \neq q^{(k_i)}$
 is equivalent to $q^{(k_{i+1})} \neq q^{(k_i)}$, we obtain from 
 \eqref{2025060122544445444P4M4A} and the inductive hypothesis that 
\begin{equation}\label{20250612411P111111111MA}
 c^{q_0, i+1}\big( q^{(k_{i+1})} \big) = \\ 
 \begin{cases} 
 1 
 & \text{if $\delta\big( q^{(k_i)}, w_{i+1} \big) \neq q^{(k_i)} $}, \\ 
 c^{q_0, i+1}\big( q^{(k_i)} \big) & \text{if $\delta\big( q^{(k_i)}, w_{i+1} \big) = q^{(k_i)} $}, 
 \end{cases} 
\end{equation}
 Moreover, from the latter case of \eqref{2025111000061020411PM1A}, 
 we find that \eqref{20250612411P111111111MA} reduces so that 
 $ c^{q_0, i+1}\big( q^{(k_{i+1})} \big) = 1$, 
 as desired. As indicated in the first case of \eqref{2025111000061020411PM1A}, we have that 
 $ c^{q_0, i+1}\big( q^{(k_{i})} \big) = 0 $ 
 if $q^{(k_i)}$ and $q^{(k_{i+1})}$ are distinct. 
 If $j \in \{ 1, 2, \ldots, |Q| \}$ is not equal to $k_i$ and is not equal to $k_{i+1}$, 
 then \eqref{2025061241888888888878P8M8A} in conjunction with the inductive hypothesis 
 give us that 
 $ c^{q_0, i+1}\big( q^{(j)} \big) = 0$, 
 as desired. The definition in \eqref{setdeltacqw} 
 together with \eqref{claimcq0i} then give us that 
\begin{equation}\label{pairwith1}
 \delta^{c}( q_0, w ) 
 = \left( q^{( k_{\ell(w)} )}, 
 1 \right), 
\end{equation}
 so that \eqref{pairwith1} and Definition \ref{reduceddefinition} together give us that 
\begin{equation}\label{pairreducedtogether}
 \overline{\delta^{c}}( q_0, w ) 
 = q^{( k_{\ell(w)} )}, 
\end{equation}
 and the right-hand side of \eqref{pairreducedtogether} is equal to 
 $\delta(q_0, w)$, according to \eqref{deltanested}. 
\end{proof}

 A source of motivation related to how Theorem \ref{theoremmainformalize} formalizes how our constructions generalize DFAs and 
 DFAOs has to do with how our constructions give rise, as described below, to a new generalization of $k$-regular sequences. 

\section{Quasi-$k$-regular sequences}
 The concept of a \emph{$k$-regular sequence} was introduced in a seminal paper by Allouche and Shallit \cite{AlloucheShallit1992} 
 as a generalization of automatic sequences inspired by the uses of automatic sequences in both number theory and formal language 
 theory. For our purposes, the following definition, out of the many equivalent definitions of a $k$-regular sequence, is appropriate. 
 Definition \ref{definitionkreg} below slightly differs from the corresponding result (Theorem 2.2) from the work of Allouche and Shallit 
 \cite{AlloucheShallit1992} (due to the context in which the expression \emph{$\mathbb{Z}$-linear combination} is used), but can be 
 shown to be equivalent to the usual definitions of the term \emph{$k$-regular sequence} (for the case of integer sequences). 

\begin{definition}\label{definitionkreg}
 A sequence $(s(n) : n \in \mathbb{N}_{0} )$ is said to be \emph{$k$-regular} if there is a nonnegative integer $E$ such that, for every 
 $e_j > E$ and for all $r_j$ such that $0 \leq r_j \leq k^{e_j} - 1$, the $n^{\text{th}}$ entry of any subsequence of $s$ of the form $ 
 (s(k^{e_j} n + r_j ) : n \in \mathbb{N}_{0} )$ may be written as an $\mathbb{Z}$-linear combination of expressions of the form $ s\left( 
 k^{f_{i, j}} n + b_{i, j} \right) $ and of a constant integer sequence, and where $f_{i, j} \leq E$ and $0 \leq b_{i, j} \leq k^{f_{i, j}} - 1$ 
 (cf.\ \cite[Theorem 16.1.3]{AlloucheShallit2003} \cite{AlloucheShallit1992}). 
\end{definition}

\begin{example}
 As in the standard text on automatic sequences \cite[p.\ 438]{AlloucheShallit2003}, we highlight a prototypical instance of a 
 $k$-regular sequence, by setting $s_{2}(n)$ as the number of $1$-digits in the base-$2$ expansion of $n \in \mathbb{N}_{0}$, so that 
 $s_{2}(2n) = s_{2}(n)$ and $s_{2}(2n + 1) = s_{2}(n) + 1$. Observe that in this latter case, we would take a $\mathbb{Z}$-linear 
 combination of $(s_2(n) : n \in \mathbb{N}_{0})$ and the constant integer sequence $(1 : n \in \mathbb{N}_{0})$. 
\end{example}

 Our investigations based on the concept of a DDFAO have led us to introduce and formulate a generalization of Definition 
 \ref{definitionkreg}, and this is illustrated below. 

\begin{example}\label{exeseq}
 Let $(Q, \Sigma, \delta, q_0, P, \Delta, \tau)$ be the DDFAO 
 illustrated in Example \ref{extuplecharge}. 
 We not enforce the assigned values such that $q_{0} = q_{1} = q_{2} = q_{3} = 1$. 
 We then define a sequence $(d(n) : n \in \mathbb{N}_{0} )$ if 
 reduced final charges associated with the given DDFAO, 
 by setting $d(n)$ as the reduced final charge 
 $ \overline{\delta^{c}}( q_0, w )$, where $w$ is the word given by the base-$2$ expansion of $n$. 
 This gives us that $d(0) = \frac{1}{2}$ and, letting nonnegative integers be denoted 
 as binary words according to base-$2$ expansions, that 
\begin{align}
 d(11s) & = \frac{3}{4}, \label{drelation1} \\ 
 d\big( 1\underbrace{00\cdots{0}}_{\ell \geq 0} \big) & = \frac{1}{2^{\ell + 1}}, \label{drelation1} \\ 
 d\big( 1\underbrace{00\cdots{0}}_{\ell \geq 1}1s \big) & = 1 - \frac{1}{2^{\ell + 2}}, \label{drelation3}
\end{align}
 letting $s$ denote a possibly empty string or word, giving rise to the rational sequence 
 $$ \left( d(n) : n \in \mathbb{N}_{0} \right) 
 = \Big( \frac{1}{2},\frac{1}{2},\frac{1}{4},\frac{3}{4},\frac{1}{8},\frac{7}{8},\frac{3}{4},\frac{3}{4},\frac{1}{16},\frac{15}{16},\frac{7}{8},
\frac{7}{8},\frac{3}{4},\frac{3}{4},\frac{3}{4},\frac{3}{4},\frac{1}{32}, \ldots \Big). $$ 
 Now, we form an integer sequence 
 $(e(n) : n \in \mathbb{N}_{0})$ from $(d(n) : n \in \mathbb{N}_{0})$ 
 by multiplying each entry in this latter sequence by a fixed power of $2$, so that 
 $e(0) = 1$ and so that 
\begin{align*}
 e(11s) & = 4 \, d(11s), \\ 
 e\big( 1\underbrace{00\cdots{0}}_{\ell \geq 0} \big) & = 
 2^{\ell + 1} d\big( 1\underbrace{00\cdots{0}}_{\ell \geq 0} \big), \\
 e\big( 1\underbrace{00\cdots{0}}_{\ell \geq 1}1s \big) 
 & = 2^{\ell + 2} d\big( 1\underbrace{00\cdots{0}}_{\ell \geq 1}1s \big), 
\end{align*}
 yielding the integer sequence
\begin{equation}\label{numericaleseq}
 (e(n) : n \in \mathbb{N}_{0}) 
 = \big( 1, 1, 1, 3, 1, 7, 3, 3, 1, 15, 7, 7, 3, 3, 3, 3, 1, 31, 15, \ldots \big). 
\end{equation}
 From the relations for the $d$-sequence among \eqref{drelation1}--\eqref{drelation3}, we obtain that 
\begin{equation}\label{e2nen}
 e(2n) = e(n), 
\end{equation}
 which recalls the definition of a $k$-regular sequence. In contrast, we have that 
\begin{equation}\label{e4nplus1}
 e(4n+1) \in \{ e(2n), \, 2 e(2n+1) + 1 \} 
\end{equation}
 for $n \in \mathbb{N}_{0}$  and that 
\begin{equation}\label{e4np3}
 e(4n+3) \in \{ e(2n), \, e(2n+1) \} 
\end{equation}
 for $n \in \mathbb{N}$. 
 Moreover, it can be shown that it is not the case that 
 $e(4 n + 1) = e(n)$ for all sufficiently large $n$ and it is not the case that
 $e(4n+1) = 2 e(2n+1) + 1$ for all sufficiently large $n$, 
 and similarly with respect to \eqref{e4np3}. 
\end{example}

 The integer sequence in \eqref{numericaleseq} obtained from our DDFAO construction in 
 Example \ref{extuplecharge} suggests a natural extension of $k$-regular sequences, 
 as below, and in view of the relations among 
 \eqref{e2nen}, \eqref{e4nplus1}, and \eqref{e4np3}. 
 To the best of our knowledge, the below generalization of $k$-regular sequences 
 has not been considered (in any equivalent form) in any previous literature on 
 variants or extensions of the concepts of $k$-automaticity and $k$-regularity for infinite sequences 
 \cite{DressvonHaeseler2003,Konieczny2024,RowlandYassawi2017}. The notion of quasi $k$-regularity introduced below can be generalized 
 to non-integer sequences, 
 but, for the sake of clarity, we restrict our attention, for the time being, to integer sequences. 

\begin{definition}\label{definitionquasi}
 We define a \emph{quasi-$k$-regular (integer) sequence}  as an integer sequence $(s(n) : n \in \mathbb{N}_{0} )$ such that there is a 
 nonnegative integer $E$ such that,  for every $e_j > E$ and for all $r_j$ such that $0 \leq r_j \leq k^{e_j} - 1$, and for fixed $m$ and 
 for $n \geq m$, the $n^{\text{th}}$ entry  of any subsequence of $s$ of the form $(s(k^{e_j} n + r_j ) : n \in \mathbb{N}_{0} )$ may be 
 written as follows. For $e_j$ and $r_j$ as specified,  let $\mathcal{I}_{e_j, r_j}$ denote a finite set of consecutive, positive integers 
 starting with $1$.  Then, for all $n \geq m$, there exists an index $\mathcal{J}_{n} \in \mathcal{I}_{e_j, r_j}$  and integer coefficients $c_{i, 
  j}^{\left( \mathcal{J}_{n} \right)}$ and  $c_{\text{constant}}^{\left( \mathcal{J}_{n} \right)}$ (possibly depending on $n$)  together with 
 values $f_{i, j}^{\left( \mathcal{J}_{n} \right)} $ and $ b_{i, j}^{\left( \mathcal{J}_{n}  \right)} $ (again possibly depending on $n$) such that 
 $0 \leq f_{i, j}^{\left( \mathcal{J}_{n} \right)} \leq E$ and $0 \leq b_{i, j}^{\left( \mathcal{J}_{n} \right)} \leq k^{f_{i, j}^{\left( \mathcal{J}_{n} 
 \right)}} - 1$ and such that 
\begin{equation*}
 s(k^{e_j} n + r_j ) = c_{\text{constant}}^{\left( \mathcal{J}_{n} \right)}
 + \sum_{i} c_{i, j}^{\left( \mathcal{J}_{n} \right) } \, s\big( k^{f_{i, j}^{\left( \mathcal{J}_{n} \right)}} n 
 + b_{i, j}^{\left( \mathcal{J}_{n} \right)} \big). 
\end{equation*}
\end{definition}

\begin{example}\label{exmathcalT}
 Set $\mathcal{T}(0) = 0$ and, for $n \in \mathbb{N}_{0}$ set 
\[ \mathcal{T}(2n) = \begin{cases} 
 1 - \mathcal{T}(n) & \text{if $n$ is a prime number,} \\
 \mathcal{T}(n) & \text{otherwise}, 
 \end{cases} \] 
 and set $$ \mathcal{T}(2n+1) = 1 - \mathcal{T}(n), $$ yielding the integer sequence $$ \big( \mathcal{T}(n) : n \in \mathbb{N}_{0} 
 \big) = (0, 1, 1, 0, 0, 0, 1, 1, 0, 1, 1, 1, 1, 0, 0, 0, 0, 1, 1, 0, 1, 0, 0, \ldots), $$ which is not currently included in the OEIS. Since 
\begin{equation}\label{mathcalevenin}
 \mathcal{T}(2n) \in \{ 1 - \mathcal{T}(n), \mathcal{T}(n) \} 
\end{equation}
 for all $n \in \mathbb{N}_{0}$, and similarly for the case of odd arguments, 
 this can be used to show that the $\mathcal{T}$-sequence satisfies 
 the conditions in Definition \ref{definitionquasi}. 
 It can also be shown that the $\mathcal{T}$-sequence is not $k$-regular. 
\end{example}

 The connection between our notion of quasi-$k$-regularity and our DDFAO construction is suggested below. 

\begin{example}\label{showequasi}
 The $e$-sequence indicated in \eqref{numericaleseq} 
 and determined by the DDFAO in Example \ref{extuplecharge} 
 is quasi-$4$-regular, and this follows from the relations among 
 \eqref{e2nen}, \eqref{e4nplus1}, and \eqref{e4np3}. 
\end{example}

 Informally, the notion of quasi-$k$-regularity naturally generalizes that for 
 $k$-regularity in the sense that quasi-$k$-regularity allows 
 ``options'' in terms of how 
 an expression of the form $ s(k^{e_j} n + r_j ) $ can be reduced 
 in terms of a preceding sequence entry, as suggested by \eqref{mathcalevenin}. 
 This can be formalized via the following result. 

\begin{theorem}\label{quasiktheorem}
 Let $(s(n) : n \in \mathbb{N}_{0})$ be a quasi-$k$-regular sequence 
 such that, in the notation of Definition \ref{definitionquasi}, we have that $m = 0$
 and every index set of the form $\mathcal{I}_{e_{j}, r_{j}}$ is a singleton set.
 Then $(s(n) : n \in \mathbb{N}_{0})$ is a $k$-regular sequence. 
\end{theorem}

\begin{proof}
 Since each index set of the form $\mathcal{I}_{e_{j}, r_{j}}$ is a singleton set, the expressions $c_{i, j}^{\left( \mathcal{J}_{n} \right)}$ 
 and $c_{\text{constant}}^{\left( \mathcal{J}_{n} \right)}$ and $f_{i, j}^{\left( \mathcal{J}_{n} \right)} $ and $ b_{i, 
 j}^{\left( \mathcal{J}_{n} \right)} $ do not depend on $n$ (or the unique element $\mathcal{J}_{n}$ in $ \mathcal{I}_{e_j, r_j}$), and, 
 hence, we may unambiguously write $c_{i, j}^{\left( \mathcal{J}_{n} \right)} = c_{i, j}$ and $c_{\text{constant}}^{\left( \mathcal{J}_{n} 
 \right)} = c_{\text{constant}}$ and $f_{i, j}^{\left( \mathcal{J}_{n} \right)} = f_{i, j}$ and $ b_{i, j}^{\left( \mathcal{J}_{n} \right)} = b_{i, j}$. 
 The expansion 
\begin{equation*}
 s(k^{e_j} n + r_j ) = c_{\text{constant}} + \sum_{i} c_{i, j} \, s\big( k^{f_{i, j}} n + b_{i, j} \big) 
\end{equation*}
 thus holds for all $n \geq 0$, so that the required conditions for a $k$-regular sequence are satisfied. 
\end{proof}

 It can be shown that the set of all quasi-$k$-regular sequences
 forms a ring structure generalizing that of the ring of $k$-regular sequences. 
 This, together with 
 Example \ref{exmathcalT}, Example \ref{showequasi}, and 
 Theorem \ref{quasiktheorem} 
 suggest that quasi-$k$-regular sequences could provide a far-reaching generalization
 of $k$-regular sequences, thus motivating our DDFAO and DDFA constructions 
 and the future subjects of research given in the concluding section below. 

\section{Conclusion}
 We conclude by encouraging the pursuit of the two research subjects below. 

 To begin with, we encourage a full exploration of quasi-$k$-regular sequences, 
 through the exploration of how Allouche and Shallit's work on $k$-regular sequences
 \cite[\S16]{AlloucheShallit2003} \cite{AlloucheShallit1992} could be translated so as to be applicable using 
 Definition \ref{definitionquasi}, and through the exploration 
 of how many further applications of $k$-regular sequences could be built upon 
 using quasi-$k$-regular sequences. 

 Also, we encourage the exploration of the following question, 
 from number-theoretic perspectives and in regard to potential applications in formal language theory:
 What is the relationship between quasi-$k$-regular sequences and the 
 sequence of values for the reduced charge-extended transition function 
 for a DDFAO (given by inputting the expansions of consecutive integers in an appropriate base)? 
 Example \ref{exeseq} suggests that by taking such a sequence, 
 it would be possible to multiply each term by a certain power of a specified base 
 to form a quasi-$k$-regular sequence. We leave it as an open problem to prove this. 

\subsection*{Acknowledgements}
 The author was supported by a Killam Postdoctoral Fellowship from the Killam Trusts and thanks Karl Dilcher
 for useful discussions.

 \ 

{\textsc{John M. Campbell}} 

\vspace{0.1in}

Department of Mathematics and Statistics

Dalhousie University

6299 South St, Halifax, NS B3H 4R2

\vspace{0.1in}

{\tt jh241966@dal.ca}


\begin{thebibliography}{99}

\bibitem{AlloucheShallit2003}%%%%%%%%%%%%%%%
 \textsc{J.-P.\ Allouche and J.\ Shallit}, 
 Automatic sequences, 
 Cambridge University Press, Cambridge (2003). 
 
\bibitem{AlloucheShallit1992}%%%%%%%%%%%%%%%q
 \textsc{J.-P.\ Allouche and J.\ Shallit}, 
 The ring of {$k$}-regular sequences, 
 \emph{Theoret.\ Comput.\ Sci.} {\bf 98(2)} (1992), 163--197. 

\bibitem{AlloucheShallit1999}%%%%%%%%%%%%%%%qq
 \textsc{J.-P.\ Allouche and J.\ Shallit}, 
 The ubiquitous {P}rouhet-{T}hue-{M}orse sequence, 
 \emph{Sequences and their applications ({S}ingapore, 1998)} Springer Ser. Discrete Math. Theor. Comput. Sci.
 Springer, London (1999), 1--16. 

\bibitem{BrandenburgSkodinis2005}%%%%%%%%%%%%%%%q
 \textsc{F.\ J.\ Brandenburg and K.\ Skodinis}, 
 Finite graph automata for linear and boundary graph languages, 
 \emph{Theoret.\ Comput.\ Sci.} {\bf 332(1-3)} (2005), 199--232. 

\bibitem{BrzozowskiMcCluskey1963}%%%%%%%%%%%%%%%
 \textsc{J.\ A.\ Brzozowski and E.\ J.\ McCluskey}, 
 Signal flow graph techniques for sequential circuit state diagrams, 
 \emph{IEEE Trans.\ Comput.} {\bf EC-12(2)} (1963), 67--76. 

\bibitem{CranstonWest2017}%%%%%%%%%%%%%%%q
 \textsc{D.\ W.\ Cranston and D.\ B.\ West}, 
 An introduction to the discharging method via graph coloring, 
 \emph{Discrete Math.} {\bf 340(4)} (2017), 766--793. 

\bibitem{DemongeotWaku2009}%%%%%%%%%%%%%%%q
 \textsc{J.\ Demongeot and J.\ Waku}, 
 Application of interval iterations to the entrainment problem in respiratory physiology, 
 \emph{Philos. Trans. R. Soc. Lond. Ser. A Math. Phys.\ Eng.\ Sci.} {\bf 367(1908)} (2009), 4717--4739. 

\bibitem{DemongeotWaku2005}%%%%%%%%%%%%%%%q
 \textsc{J.\ Demongeot and J.\ Waku}, 
 Counter-examples about lower- and upper-bounded population growth, 
 \emph{Math.\ Popul. Stud.} {\bf 12(4)} (2005), 199--209. 

\bibitem{DressvonHaeseler2003}
 \textsc{A.\ W.\ M.\ Dress and F.\ von Haeseler}, 
 A semigroup approach to automaticity, 
 \emph{Ann.\ Comb.} {\bf 7(2)} (2003), 171--190. 

\bibitem{Eilenberg1974}%%%%%%%%%%%%%%%q
 \textsc{S. Eilenberg}, 
 Automata, languages, and machines. {V}ol. {A}, 
 Academic Press [Harcourt Brace Jovanovich, Publishers], New York (1974). 

\bibitem{FraigniaudIlcinkasPeerPelcPeleg2005}%%%%%%%%%%%%%%%
 \textsc{P.\ Fraigniaud, D.\ Ilcinkas, 
 G.\ Peer, A.\ Pelc, and D.\ Peleg}, 
 Graph exploration by a finite automaton, 
 \emph{Theoret.\ Comput. Sci.} {\bf 345(2-3)} (2005), 331--344. 

\bibitem{GiammarresiMontalbano1999}%%%%%%%%%%%%%%%
 \textsc{D.\ Giammarresi and R.\ Montalbano}, 
 Deterministic generalized automata, 
 \emph{Theoret.\ Comput.\ Sci.} {\bf 215(1-2)} (1999), 191--208. 

\bibitem{Gruber2012}%%%%%%%%%%%%%%%q
 \textsc{H.\ Gruber}, 
 Digraph complexity measures and applications in formal language theory, 
 \emph{Discrete Math.\ Theor.\ Comput. Sci.} {\bf 14(2)} (2012), 189--204. 

\bibitem{HanWood2005}%%%%%%%%%%%%%%%
 \textsc{Y.-S.\ Han and D.\ Wood}, 
 The generalization of generalized automata: expression automata, 
 \emph{Internat. J. Found. Comput. Sci.} {\bf 16(3)} (2005), 499--510. 

\bibitem{Hashigushi1991}%%%%%%%%%%%%%%%q
 \textsc{K.\ Hashigushi}, 
 Algorithms for determining the smallest number of nonterminals (states) sufficient for generating (accepting) a regular language, 
 \emph{Proc. ICALP'91} Lecture Notes in Computer Science {\bf 510} 
 Springer, Berlin (1991), 641--648. 

\bibitem{Kelarev2004}%%%%%%%%%%%%%%%q
 \textsc{A.\ V.\ Kelarev}, 
 Labelled {C}ayley graphs and minimal automata, 
 \emph{Australas.\ J.\ Combin.} {\bf 30} (2004), 95--101. 

\bibitem{Konieczny2024}%%%%%%%%%%%%%%%q
 \textsc{J.\ Konieczny}, 
 On asymptotically automatic sequences, 
 \emph{Acta Arith.} {\bf 215(3)} (2024), 249--287. 

\bibitem{KonitzerSimon2017}%%%%%%%%%%%%%%%
 \textsc{M.\ Konitzer and H.\ U.\ Simon}, 
 Regular languages viewed from a graph-theoretic perspective, 
 \emph{Inform. and Comput.} {\bf 253} (2017), 484--496. 

\bibitem{Nagy2023}%%%%%%%%%%%%%%%
 \textsc{B.\ Nagy}, 
 State-deterministic finite automata with translucent letters
 and finite automata with nondeterministically translucent letters, 
 \emph{Proceedings of the 16th {I}nternational {C}onference on
 {A}utomata and {F}ormal {L}anguages} 
 Electron. Proc. Theor. Comput. Sci. (EPTCS) {\bf 386} 
 EPTCS (2023), 170--184. 

\bibitem{RestivoVaglica2012}%%%%%%%%%%%%%%%q
 \textsc{A.\ Restivo and R.\ Vaglica}, 
 A graph theoretic approach to automata minimality, 
 \emph{Theoret.\ Comput.\ Sci.} {\bf 429} (2012), 282--291. 

\bibitem{RowlandYassawi2017}%%%%%%%%%%%%%%%
 \textsc{E.\ Rowland and R.\ Yassawi}, 
 Profinite automata, 
 \emph{Adv.\ in Appl.\ Math.} {\bf 85} (2017), 60--83. 

\bibitem{RozenbergSalomaa1997}%%%%%%%%%%%%%%%q
 \textsc{G.\ Rozenberg and A.\ Salomaa}, 
 Handbook of formal languages. {V}ol. 1, 
 Springer-Verlag, Berlin (1997). 

\bibitem{Shallit2009}%%%%%%%%%%%%%%%q
 \textsc{J. Shallit}, 
 A second course in formal languages and automata theory, 
 Cambridge University Press, Cambridge (2009). 

\bibitem{Sloane2025}%%%%%%%%%%%%%%%q
 \textsc{N.\ J.\ A.\ Sloane et~al.}, 
 The {O}n-{L}ine {E}ncyclopedia of {I}nteger {S}equences, 2025. 
 Available at \url{https://oeis.org}. 

\bibitem{Wood1987}%%%%%%%%%%%%%%%
 \textsc{D.\ Wood}, 
 Theory of computation, 
 Harper \& Row, Publishers, New York (1987). 

\end{thebibliography}
\end{document}